\documentclass[twocolumn,preprintnumbers,superscriptaddress]{revtex4-2}
\usepackage{amsmath}
\usepackage{amssymb}
\usepackage{graphicx}
\usepackage{hyperref}
\usepackage{physics}
\usepackage{xcolor}
\usepackage{xspace,slashed}
\input{qcircuit.sty}
\usepackage{multirow}
\usepackage[utf8]{inputenc}
\usepackage{comment}
\usepackage{amsthm}
\usepackage{lipsum}

\newtheorem{theorem}{Theorem}
\newtheorem{definition}{Definition}
\newtheorem{lemma}{Lemma}

\newcommand{\beq}{\begin{equation}}
\newcommand{\eeq}{\end{equation}}

\newcommand{\relu}{{\rm ReLU}}
\newcommand{\poly}{{\rm poly}}
\newcommand{\step}{{\rm step}}

\newcommand{\F}{\mathcal{F}}

\begin{document}

\title{One qubit as a Universal Approximant}

\newcommand{\BSC}{Barcelona Supercomputing Center, Barcelona 08034, Spain}
\newcommand{\ICCUB}{Departament de F\'isica Qu\`antica i Astrof\'isica and Institut de Ci\`encies del Cosmos, Universitat de Barcelona, Barcelona 08028, Spain}
\newcommand{\IFAE}{Institut de F\'\i sica d'Altes Energies, The Barcelona Institute of Science and Technology, Bellaterra 08193, Spain}
\newcommand{\Qilimanjaro}{Qilimanjaro Quantum Tech, Barcelona 08007, Spain}
\newcommand{\CQT}{Center for Quantum Technologies, Singapore 117543, Singapore}
\newcommand{\TII}{Quantum Research Centre, Technology Innovation Institute, P.O.Box: 9639, Abu Dhabi, United Arab Emirates}

\author{Adrián Pérez-Salinas}
\affiliation{\BSC}
\affiliation{\ICCUB}

\author{David L\'opez-N\'u\~nez}
\affiliation{\BSC}
\affiliation{\IFAE}
\affiliation{\ICCUB}

\author{Artur Garc\'\i a-S\'aez}
\affiliation{\BSC}
\affiliation{\Qilimanjaro}

\author{P. Forn-D\'\i az}
\affiliation{\IFAE}
\affiliation{\Qilimanjaro}

\author{Jos\'{e} I. Latorre}
\affiliation{\Qilimanjaro}
\affiliation{\CQT}
\affiliation{\TII}

\begin{abstract}
A single-qubit circuit can approximate any bounded complex function stored in the degrees of freedom defining its quantum gates.
The single-qubit approximant presented in this work is operated through a series of gates that take as their parameterization the independent variable of the target function and an additional set of adjustable parameters. The independent variable is re-uploaded in every gate while the parameters are optimized for each target function. The output state of this quantum circuit becomes more accurate as the number of re-uploadings of the independent variable increases, i. e., as more layers of gates parameterized with the independent variable are applied. In this work, we provide two different proofs of this claim related to both Fourier series and the Universal Approximation Theorem for Neural Networks, and we benchmark both methods against their classical counterparts. We further implement a single-qubit approximant in a real superconducting qubit device, demonstrating how the ability to describe a set of functions improves with the depth of the quantum circuit.
This work shows the robustness of the re-uploading technique on Quantum Machine Learning. 
\end{abstract}

\maketitle
\section{Introduction}
A quantum computer can be viewed as a machine that receives inputs and delivers outputs through the read-out of qubits. The design of the sequence of quantum gates forming the circuit will determine the kind of processing performed. A fundamental question to pose is whether a quantum circuit can deliver any possible functionality and, if so, what number of qubits and depth are required to achieve a given accuracy.

This problem is reminiscent of a series of classical theorems that establish that a given function can be re-expressed as a linear combination of other specific functions
~\cite{fourier-dirichlet1829, fourier-riemann1867, uat-cybenko1989, uat-hornik1991}. 
In particular, in classical machine learning the Universal Approximation Theorem (UAT) ~\cite{uat-cybenko1989} proves that a neural network with a unique intermediate hidden layer can converge to approximate any continuous function. The accuracy of the approximation increases with the number of neurons in the intermediate layer. It is important to notice that each of these neurons is fed with the original data of the problem. The query complexity of the process increases linearly with the number of neurons. This observation is critical to find out an equivalent result in a quantum formulation in order to support progress in  quantum machine learning  \cite{reuploading-perezsalinas2020,proton-perezsalinas2020, autoencoders-bravoprieto2020,
mitarai:2018circuit,Zhu:2019circuit,lloyd:2020embeddings, liu:2020svm,Rebentrost:2014svm,lloyd:2013ml}. Previous works have already made relevant contributions in this area~\cite{encoding-schuld2020, universal:goto2020}, in particular assessing the universal expressive power of quantum circuits.

In this paper, we present two independent proofs that any bounded complex function can be approximated in a convergent way by a quantum circuit acting on one qubit, constituting a single-qubit universal approximant. This demonstrates the precise representation power of a single-qubit circuit, which increases as more layers are added. The essential element of the present construction is the re-uploading of the input variables along the action of the quantum gates~\cite{reuploading-perezsalinas2020}. Thus, in analogy to neural networks, query complexity is attached to accuracy. The first way to prove this result is to make contact with harmonic analysis. This is a natural step as single-qubit gates are expandable in Fourier series that can be rearranged to fit existing theorems. The second method is analogous to the UAT using a translation into quantum circuits. A series of specific gates leads to an output state that approximates functions uniformly. In both cases, the quantum theorems inherit the applicability and characteristics of their classical counterparts.

The practical way to approximate any function with a quantum circuit requires finding an explicit set of parameters to define the unitary gates. This can be accomplished in a variational way. Compared to Fourier series, this approach brings more power to a quantum computer in practise. The possibility of taking angles which are free of being multiples of a given basic frequency provides a larger representation capability to a quantum circuit. However, this is analogue to neural networks with weights which are not constrained to take specific values. 

We provide numerical benchmarks of these theorems computed via classical simulation of quantum computers. Our simulations show how an increasing query complexity can improve the accuracy of the approximation of a number of test functions. The way a single-qubit circuit can approximate any function in an experimental setup by using a superconducting qubit is explicitly illustrated. Experimental results confirms the trend of simulations up to a point where the accumulation of errors dominates the experiment.

This procedure of storing complex functions in single-qubit circuits acquires its importance when addressing Machine-Learning problems. The goal of any Machine Learning model is to extract and generalize the hidden features of some training dataset in order to predict the behaviour of unseen data. That is, the model must learn some unknown function from sampling data while the structure of such function remains covert. The proofs given in this work ensure that this kind of models can learn any function underlying the training data. In addition, the processing of data in Machine Learning models must be performed in a linear way to avoid the emergence of biases, as it is done in this case. Essentially, the theoretical work proposed in this work aims to play for quantum circuits the same role UAT plays for Neural Networks. 

It is worth mentioning that, even though one can store a complex function in a single-qubit circuit, retrieving that information from the quantum state is costly and requires a large amount of measurements. Instead, this algorithm should be regarded as a subroutine to be included in a more complex computation. For instance, this subroutine could play the role of a classifier~\cite{reuploading-perezsalinas2020} or an approximant of unknown functions~\cite{proton-perezsalinas2020}. These works also suggest the use of similar approaches extended to circuits with many qubits.

This article is organized as follows. In Sec.~\ref{sec:theory} we introduce the idea of a single-qubit approximant and present two theorems on its universality. Sec.~\ref{sec:benchmark} is devoted to the numerical benchmarks to test the approximation algorithms. The experimental implementation using a superconducting qubit is described in Sec.~\ref{sec:experiment}. Results are presented in Sec.~\ref{sec:results}. We leave conclusions for Sec.~\ref{sec:conclusions}.
More details on the results of this work can be found in the appendix.

\section{Universality of the single-qubit approximant}\label{sec:theory}

In this section, we propose an encoding of mathematical functions as the degrees of freedom of a single qubit state. We also define two different circuit architectures that approximate those functions and present theorems supporting this claim. 

\begin{figure*}
\centering
    \includegraphics[width=.7\linewidth]{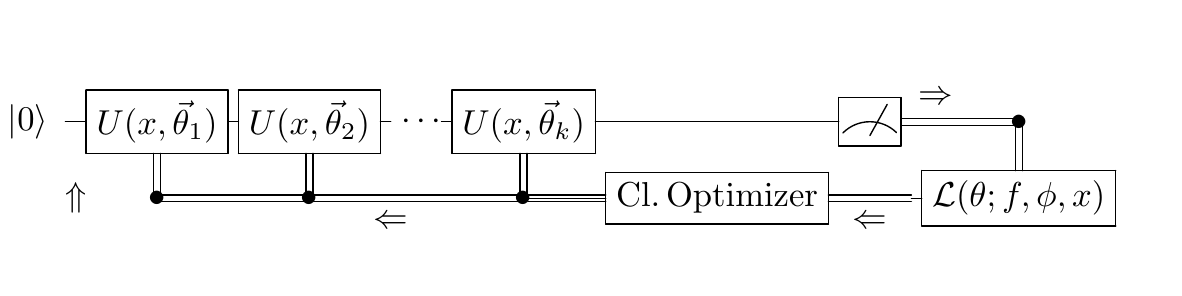}
    \caption{Scheme for the hybrid algorithm. The gates $U(x, \vec\theta_i)$ define the operation performed by the quantum circuit. All $\vec\theta_i$ are independent of each other. Using the measured output state, a loss function $\mathcal L$ is constructed using these measurements. A classical optimizer looks for the set of parameters minimizing $\mathcal L$.}
    \label{fig:var_circuit}
\end{figure*}

\subsection{Set-up of the problem}

The most general representation of a single-qubit quantum state stores a single complex number. That is, 
\begin{equation}
    \ket{\psi} = \sqrt{1 - f^2} \ket 0 + f e^{i \phi} \ket 1, 
\end{equation}
with $f, \phi$ real numbers and $f \in [0, 1]$, $\phi \in [0, 2\pi)$. Our aim is to encode a complex function within the values ($f, \phi$) by defining them as $ f:\mathbb{R}^m \rightarrow [0, 1]$ and $\phi: \mathbb{R}^m \rightarrow [0, 2\pi)$. To do so, we
design a circuit $\mathcal U_{f, \phi}(x)$ such that its output state approximates the desired complex function as
\begin{equation}\label{eq:aprox}
    \bra{1}\mathcal{U}^{(k), s}_{f, \phi}\ket{0} \rightarrow f(x) e^{i \phi(x)}
\end{equation}
Note that building an approximation to a bounded complex function is enough to address any bounded complex function by a simple shifting and re-scaling. In addition, approximating a complex function includes real-valued functions by either setting $\phi(x) = 0$ or relating the real-valued function to the modulus of other complex functions.

\begin{definition} \label{def:prod_general}
The $k$-th approximating circuit is defined as 
\begin{equation} \label{eq:prod_general}
    \mathcal{U}^{(k), s}_{f, \phi} = \prod_{i = 1}^k U^s(x, \vec\theta_i),
\end{equation}
where $U^s(x, \vec\theta)$ is a fundamental gate depending on $x$ and a set of parameters $\vec\theta$. $s$ stands for the type of single-qubit gate used in this work. 
\end{definition}

The models chosen in our construction will be made explicit later, including the exact definition of the so-called {\sl fundamental gate}. The expected behavior of this quantity $\mathcal{U}^{(k), s}_{f, \phi}$ is that the approximation from Eq.~\eqref{eq:aprox} will improve as the number $k$ increases, that is, as the independent variable is re-uploaded multiple times. As we shall see, the appropriate choice of these parameters $\vec\theta_i$ enables a systematic approximation of any functionality. Equivalently, the optimal values of $\vec\theta_i$ depend on $f(x)$ and $\phi(x)$.

In general, the set of parameters for a given gate
$\vec\theta_i$ is composed of a set of angles.
The quest for the optimal set of parameters $\lbrace \vec\theta_1, \vec\theta_2, \ldots, \vec\theta_k \rbrace$ is driven by optimizing a particular loss function $\mathcal{L}(\theta; f, \phi, x)$. This loss function must be designed in such a way that Eq.~\eqref{eq:aprox} becomes an equality as $\mathcal{L} \rightarrow 0$. The optimal parameters are then
\begin{equation}
    \lbrace \vec\theta_1, \vec\theta_2, \ldots, \vec\theta_k \rbrace = {\rm argmin}_\theta \mathcal{L}(\theta; f, \phi, x).
\end{equation} 

Therefore, the proposed quantum procedure for storing functions within the output state of a given circuit belongs to the class of hybrid quantum-classical variational algorithms. Variational algorithms are quantum algorithms whose global structure is defined, but the exact gates are not~\cite{vqe, qaoa}. A scheme of the proposed algorithm is depicted in Fig.~\ref{fig:var_circuit}.

\subsection{Two theorems on universality}
We complete the structure of the algorithm presented above with the design of the single-qubit gates $U^s$ aforementioned in Definition~\ref{def:prod_general}. In the following, we present two sets of single-qubit gates to construct quantum circuits that represent arbitrary complex functions. Each set is based on known results from the theory of function approximations, namely Fourier series~\cite{fourier-dirichlet1829, fourier-riemann1867} and Universal Approximation Theorem (UAT)~\cite{uat-cybenko1989, uat-hornik1991}, respectively. The range of applicability of these theorems for quantum circuits and the conditions for universality are thus inherited from their classical counterparts.

\subsubsection*{Quantum Fourier series}

Fourier series as a constructive method permits expressing a great range of target functions defined within an interval as a sum of a set of known functions.

\begin{theorem}\label{th:fourier}
{\bf Fourier series}~\cite{fourier-dirichlet1829, fourier-riemann1867, fourier-carleson1966, fourier-turan1970}\\
Let $z$ be any function $z: \mathbb{R} \rightarrow \mathbb{C}$ with a finite number of finite discontinuities integrable within an interval $[a, b] \in \mathbb{R}$ of length $P$. Then 
\begin{equation}\label{eq:fourier_exp_complex}
    z_N(x) = \sum_{n = -N}^N c_n e^{i \frac{2\pi n x}{P}},
\end{equation}
where
\begin{equation}
    c_n = \frac{1}{P} \int_P z(x) e^{-i\frac{2\pi n x}{P}} dx,
\end{equation}
approximates $z(x)$ as
\begin{equation}
    \lim_{N\rightarrow\infty} z_N(x) = z(x).
\end{equation}
\end{theorem}

Now we present an extension of Fourier series to a quantum circuit as explicited in Def.~\ref{def:prod_general}. First, we define the Fourier gate $U^\mathcal F$:

\begin{definition}\label{def:fourier_gate}
Let the fundamental Fourier gate $U^\mathcal F$ be
\begin{multline}\label{eq:unitary_f}
 U^{\mathcal{F}}(x; \underbrace{\omega, \alpha, \beta, \varphi, \lambda}_{\vec\theta}) = R_z\left(\alpha + \beta\right) R_y(2\lambda)\times\\R_z\left(\alpha - \beta\right) R_z(2\omega x) R_y(2\varphi),
\end{multline}
with $\alpha, \beta, \varphi, \lambda, \omega \in \mathbb{R}$.
\end{definition}

Intuitively, $\alpha, \beta, \varphi, \lambda$ are related to the coefficients of a single Fourier step, while $\omega$ may be identified as the corresponding frequency.
The relationship between these parameters and the original Fourier coefficients is explicitly shown in Appendix~\ref{app:fourier}.

\begin{theorem}\label{th:q_fourier}
{\bf Quantum Fourier series} \\
Let $f,\phi$ be any pair of functions $f: \mathbb{R} \rightarrow [0, 1]$ and $\phi: \mathbb{R} \rightarrow [0, 2\pi)$ , such that $z(x) = f(x)e^{i\phi(x)}$ is a complex function with a finite number of finite discontinuities integrable within an interval $[a, b] \in \mathbb{R}$ of length $P$.
Then, there exists a set of parameters $\lbrace\vec\theta_1, \vec\theta_2, \ldots, \vec\theta_N \rbrace$ such that
\begin{equation}
     \bra 1 \prod_{i=1}^N U^\mathcal{F}(x, \vec\theta_i) \ket 0 = z_N(x), 
\end{equation}
with $z_N(x)$ the $N$-terms Fourier series.
\end{theorem}

When the building blocks are the $U^\mathcal{F}(x, \theta_i)$ defined in Eq.~\eqref{eq:unitary_f}, the unitary operation as defined in Eq.\eqref{eq:prod_general} generates a total unitary gate that outputs a $N$-term Fourier series when applied to an initial state $\ket 0$. Taking $\ket 0$ as the initial state implies no loss of generality, since we can transform $\ket 0$ into any other initial state by adjusting the first $U^\mathcal{F}$. The Fourier series behavior is only achieved if all $\lbrace\vec\theta_i\rbrace$ take specific values leading to a final result that exactly matches the Fourier coefficients. However, since this procedure relies on quantum-classical variational methods, we will look for the optimal parameters by means of a classical optimizer. This freedom gives room to configurations surpassing the performance of the standard Fourier series, especially for shallow circuits. However, the recipe to construct the Fourier series by performing well-defined calculations is instead lost. For details on the proof of this theorem we refer the reader to Appendix~\ref{app:fourier}.

\subsubsection*{Quantum UAT}
The Universal Approximation Theorem (UAT) demonstrates that any continuous function of a $m$-dimensional variable can be uniformly approximated as a sum of a specific set of functions with adjustable parameters. The first formulation restricted the functions to be sigmoidal functions~\cite{uat-cybenko1989}. Later works extended the result to any non-constant bounded continuous function~\cite{uat-hornik1991}. This theorem is directly applied to neural networks containing one hidden layer. 

\begin{theorem}\label{th:UAT}
{\bf Universal Approximation Theorem}~\cite{uat-cybenko1989, uat-hornik1991, uat-leshno1993}
Let $I_m$ denote the $m$-dimensional cube $[0, 1]^m$. The space of continuous functions on $I_m$ is denoted by $C(I_m)$, and we use $|\cdot|$ to denote the uniform norm of any function in $C(I_m)$. 
Let $\sigma:\mathbb R \rightarrow \mathbb{R}$ be any non-constant bounded continuous function. 
Given a function $f \in C(I_m)$ there exists an integer $N$ and a function 
\begin{equation}\label{eq:UAT}
    G(\vec x) = \sum_{n=1}^N \alpha_n \sigma(\vec w_n \cdot \vec x + b_n),
\end{equation}
such that 
\begin{equation}
    |G(\vec x) - f(\vec x)| < \varepsilon, \qquad \forall \vec x \in I_m,
\end{equation}
for $\vec w_n\in \mathbb{R}^m$ and $b_n, \alpha_n \in \mathbb{R}$ for any $\varepsilon > 0$.
\end{theorem}

This theorem is an existence theorem, and thus it does not specify how many terms from Eq.~\eqref{eq:UAT} are needed to achieve an accuracy $\varepsilon$. 
Note that the UAT can be immediately applied to complex functions by substituting the real-valued function $\sigma(\cdot)$ with some complex-valued function. In particular, it works if $\sigma(\cdot) \rightarrow e^{i (\cdot)}$. A proof is shown in Appendix~\ref{app:real_to_complex_uat}.

We can now translate the UAT to the proposed quantum circuit by defining the fundamental single-qubit gate.

\begin{definition}
Let the fundamental UAT gate $U^{\rm UAT}$ be
\begin{equation}\label{eq:unitary_uat}
 U^{\rm UAT}(\vec x; \underbrace{\vec \omega, \alpha, \varphi}_{\vec\theta}) = R_y(2\varphi)R_z(2\vec\omega\cdot\vec x + 2\alpha), 
\end{equation}
with $\lbrace \vec\omega, \alpha, \varphi \rbrace \in \lbrace \mathbb{R}^m, \mathbb{R}, \mathbb{R}\rbrace$.
\end{definition}

Intuitively, $\vec\omega$ and $\alpha$ are equivalent to the weights and bias in a neural network, while $\varphi$ plays the role of the coefficient.

\begin{theorem}\label{th:q_UAT}
{\bf Quantum UAT} \\
Let $f,\phi$ be any pair of functions $f: I_m \rightarrow [0, 1]$ and $\phi: I_m \rightarrow [0, 2\pi)$ , such that $z(\vec x) = f(\vec x)e^{i\phi(\vec x)}$ is a complex continuous function on $I_m$, with $I_m = [0, 1]^m$. Then there is an integer $N$ and a set of parameters $\lbrace\vec\theta_1, \vec\theta_2, \ldots, \vec\theta_N \rbrace$ such that
\begin{equation}
    \left\vert f(\vec x)e^{i\phi(\vec x)} -  \bra 1 \prod_{i=1}^N U^{\rm UAT}(\vec x, \vec\theta_i) \ket{0} \right\vert < \epsilon,
\end{equation}
for any $\epsilon > 0$. 
\end{theorem}

This theorem is analogous to the classical one, and one can arrive at its proof by following the steps developed in Ref.~\cite{uat-cybenko1989}. All theorems supporting the original formulation of the UAT also hold for the quantum version. For more details on the demonstration of the quantum UAT, we refer the reader to Appendix~\ref{app:uat}.

\subsubsection*{Differences between approaches}

The quantum universality theorems proposed here inherit the range of applicability, advantages and limitations of their classical counterparts. The Fourier approach is guaranteed to work for all integrable functions with a finite number of finite discontinuities. This range of functions includes --but is not limited to-- continuous functions. The UAT only gives support to continuous functions, which is useful from a practical perspective, but less robust than the Fourier series.

The Fourier theorem holds for functions depending on a single variable. However, the extension to multi-dimensional spaces is complicated and requires a space of parameters whose size increases exponentially with the number of dimensions~\cite{fourier-riemann1867}. However, in the UAT case the use of multi-variable $\vec x$ arises naturally by adjusting the dimension of the weights.

\section{Numerical experiments}\label{sec:benchmark}

In this section, we numerically explore how the theorems explained in Sec.~\ref{sec:theory} perform in practice. We present two different kinds of benchmarks for real and complex functions, respectively. These benchmarks collect results using both $U^\mathcal{F}$ and $U^{\rm UAT}$ gates. 
Benchmarks are performed using simulations that include no decoherence, but do contain sampling uncertainty. We present simulations with up to 6 layers. 

The aim of this benchmark is to compare the results of quantum and classical methods. The classical Fourier representation can be obtained by following Theorem \ref{th:fourier}. In the UAT case, we follow the description from Theorem \ref{th:UAT} with $\sigma(\cdot)$ being a cosine for real functions and $e^{i (\cdot)}$ for complex functions. The parameters are found by employing specific classical optimization methods. For the quantum UAT case we take $H \ket 0 = \ket +$ as the initial state. 

All simulations are performed using the framework {\tt QIBO}~\cite{qibo}. The code computing the numerical experiments as well as the final results can be found on {\tt GitHub}~\cite{github}.

\subsubsection*{$Z$ benchmark for real functions}

For the first benchmark we consider a single-variable, real-valued function $-1 \leq f(x) \leq 1$ related to the observable $\langle Z \rangle \sim f(x)$. The quantum state we want to represent is then
\begin{equation}
    \ket{\psi(x)}_Z = \sqrt{\frac{1 + f(x)}{2}}\ket 0 + e^{i\phi}  \sqrt{\frac{1 - f(x)}{2}}\ket 1,
\end{equation}
where $\phi$ is a phase that in general may be $x$-dependent, but it is assumed constant at this stage. The $\chi^2$ function that drives the optimization is then
\begin{equation}
    \chi^2 = \frac{1}{M} \sum_{j=1}^M \left(\langle Z(x_j) \rangle - f(x_j) \right)^2,
\end{equation}
where $M$ is the total number of samples of $x$.

The $Z$ benchmark is first tested against four different functions of interest
\begin{eqnarray}\label{eq:functions}
    \relu(x) &=& \max(0, x),\\\label{eq:functions1}
    \tanh(a  x) & \; {\rm for} \; & a = 5, \\\label{eq:functions2}
    \step(x) &=& x / \vert x \vert;\quad 0\; {\rm if}\; x=0, \\\label{eq:functions3}
    \poly (x) &=& \vert 3 x^3 (1 - x^4)\vert. 
\end{eqnarray}
All functions are conveniently rescaled to fit the limits $-1 \leq f(x) \leq 1$. In all cases, $x \in [-1, 1]$. The $\relu(\cdot)$ and $\tanh(\cdot)$ functions are chosen given the central role they play in the field of Machine Learning. $\step(\cdot)$ presents a discontinuity, which implies a challenge in the approximation. $\poly(\cdot)$ is chosen as it contains wavy features arising from  non-trigonometric functions. 

Next, we test our approach against four functions of two variables in order to check how the quality of the approximations evolves as more dimensions are added to the problem. Those are known 2D functions named {\tt adjiman, brent, himmelblau, threehump}~\cite{2d_functions}. These functions are chosen as representatives of a variety of difficulties the algorithm needs to overcome.
In the 2D case, the functions are 
conveniently rescaled to fit the limits $-1 \leq f(x) \leq 1$ and $(x, y) \in [-5, 5]^2$. A definition of these functions can be found in Appendix~\ref{app:2D_benchmark}.

In this benchmark, both the UAT and Fourier quantum and classical methods are considered for the one-dimensional functions. However, 2D functions are only tested for UAT methods since the theorems from Sec.~\ref{sec:theory} do not support multidimensional Fourier series.

\subsubsection*{$X-Y$ benchmark for complex functions}

In order to test the performarnce of the presented algorithm for fitting complex functions, we propose a tomography-like benchmark. Since complex functions have real and imaginary parts, one needs to measure at least two observables in the qubit space. In this case, we chose the observables to be $\langle X \rangle$ and $\langle Y \rangle$ for the real and imaginary parts, that is {$\langle X \rangle + i \langle Y \rangle \sim f(x) e^{i g(x)}$}. The quantum state that permits this identification is
\begin{multline}
    \ket{\psi(x)}_{XY}  = \sqrt{\frac{1 + \sqrt{1 - f(x)}}{2}} \ket 0 \\+ e^{i g(x)}\sqrt{\frac{1 - \sqrt{1 - f(x)}}{2}} \ket 1. 
\end{multline}
It is then possible to construct a $\chi^2$ function as 
\begin{equation}
    \chi^2 = \frac{1}{M} \sum_{j=1}^M \left\vert \langle X(x) \rangle + i \langle Y(x) \rangle - f(x)e^{i g(x)} \right\vert^2.
\end{equation}

For the $X-Y$ benchmark we test the algorithm against all possible combinations of real and imaginary parts of the functions defined in Eqs.~\eqref{eq:functions}--\eqref{eq:functions3}, conveniently renormalized to ensure that $\langle X \rangle^2 + \langle Y \rangle^2 \leq 1$.

\subsection{Optimization techniques}

An optimization process is required to find the optimal gate parameters. 
For the classical methods, we use standard optimization techniques~\cite{scikit-learn}. Since the number of parameters we deal with in this problem is relatively low compared to the parameter space obtained for instance in Deep Learning, the {\tt BFGS} and {\tt L-BFGS} algorithms~\cite{bfgs, l-bfgs} as implemented by {\tt scipy}~\cite{scipy} are used. These algorithms belong to the class of gradient-based optimizers.

For quantum methods, optimization brings more problems that are yet to be solved. In particular, the landscape of the loss function in the parameter space remains unknown, and thus it is hard to infer what kind of classical optimizers perform well for each particular problem. For this reason we look for the best-fit parameters using the aforementioned {\tt L-BFGS} algorithm and the genetic option {\tt CMA}~\cite{cma, cma-package}. Genetic algorithms explore vast regions of the parameter space and do not depend on gradients. However, they usually require more function evaluations to converge to the minimum. In this case it is not possible to guarantee that the solution found by any optimization algorithm is the global minimum of our loss function.

\section{Experimental implementation of the Approximation Theorems}\label{sec:experiment}

\begin{figure*}[t!]
    \centering
    \includegraphics[width=1.1\textwidth]{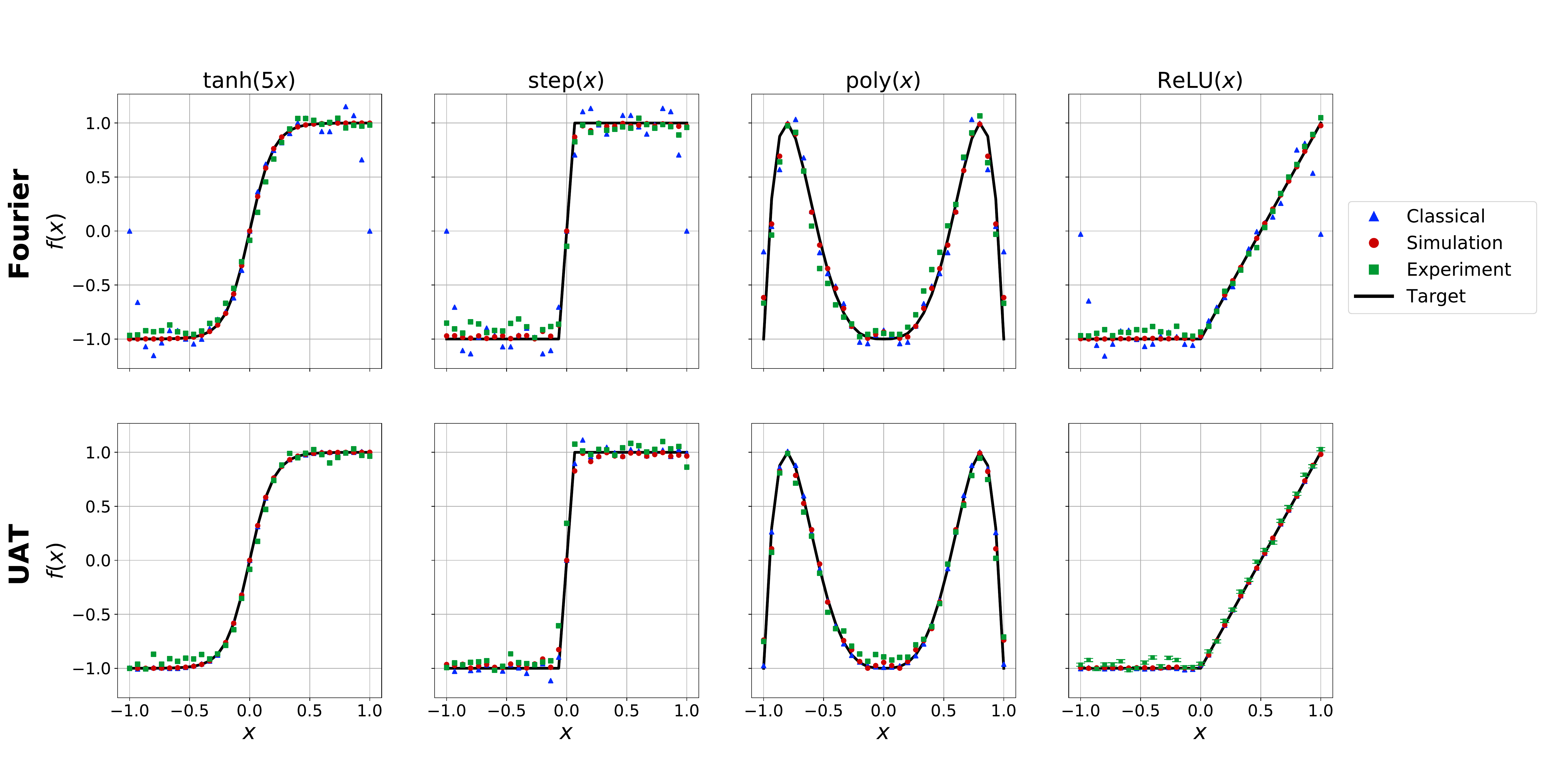}
    \caption{Fittings for four real-valued functions using the $Z$ benchmark with five layers. Blue triangles represent classical models, namely Fourier and UAT, while red dots represent its quantum counterparts computed using a classical simulator. Green squares are the experimental execution of the optimized quantum model using a superconducting qubit. The target function is plotted in black for comparison. The analysis for experimental errors is plotted for the $\relu$ function and the UAT model.}
    \label{fig:real_funs}

    \includegraphics[width=\linewidth]{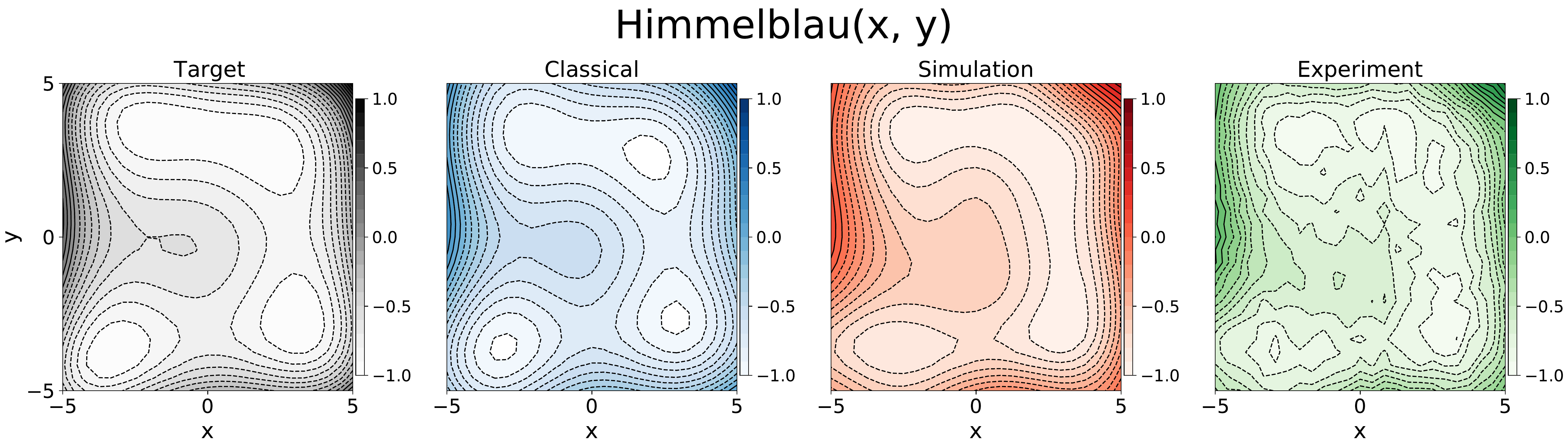}
    \caption{Fittings for the 2D function Himmelblau properly normalized using the $Z$ benchmark for five layers. The blue plot represents the classical UAT model, while the red plot represents its quantum counterpart simulated. The green plot is the experimental execution of the optimized quantum model. The target function is painted in black. In all drawings, the lines corresponds to the same levels in the $Z$ axis.}
    \label{fig:2d_himmelblau}
\end{figure*}

We implement the single-qubit universal approximant in a superconducting qubit circuit cooled to the base temperature of a dilution refrigerator (20mK). The qubit is a 3D transmon geometry~\cite{3d-transmon} located inside an aluminum three-dimensional cavity. The cavity bare frequency, $\omega_c = 2\pi \times 7.89$ GHz, is greatly detuned from the qubit frequency, $\omega_q = 2\pi \times 4.81 $ GHz. Hence, there is a qubit state-dependent dispersive shift on the cavity resonance, $2\vert\chi\vert = 2\pi \times 1.5 $MHz. The qubit anharmonicity is $\alpha = -2\pi \times 324$ MHz and the qubit relaxation and spin-echo decay times are, respectively, $T_1 = 15.6~\mu s$ and $T_{2Echo} = 12.0~\mu s$. These time scales exceed the operation times needed to implement the algorithm up to 6 layers by 2 orders of magnitude. Additional information on the experimental methods can be found in App. \ref{app:experiment}.

In order to implement the gate sequences defined in the previous section we follow the correspondence between logical and physical gates as shown in Fig.~\ref{fig:coh-times}c in App. \ref{app:experiment}. The phase of each pulse is selected at the pulse generator to modify the rotation axis, producing either X or Y rotations as required. The Z rotations are, in turn, virtual~\cite{virtual-zgates}. The microwave pulses incorporate a DRAG correction~\cite{drag-a,drag-b} which leads to an error per gate $\epsilon = 0.01$ found with randomized benchmarking~\cite{randomized-benchmarking}. Randomized benchmarking measures errors in Clifford gates and not arbitrary angle rotations, which are instead used in this experiment, yet offers a reasonable estimate on the overall fidelity of our gates. The gate error observed is probably limited due to a non-ideal filtering of the measurement lines in the fridge. In order to achieve better qubit state readout visibility and shorter operation times, a reset protocol is applied prior to the main sequence~\cite{reset-protocol}. 

\section{Results}\label{sec:results}

\begin{figure}[t!]
    \centering
    \includegraphics[width=.8\linewidth]{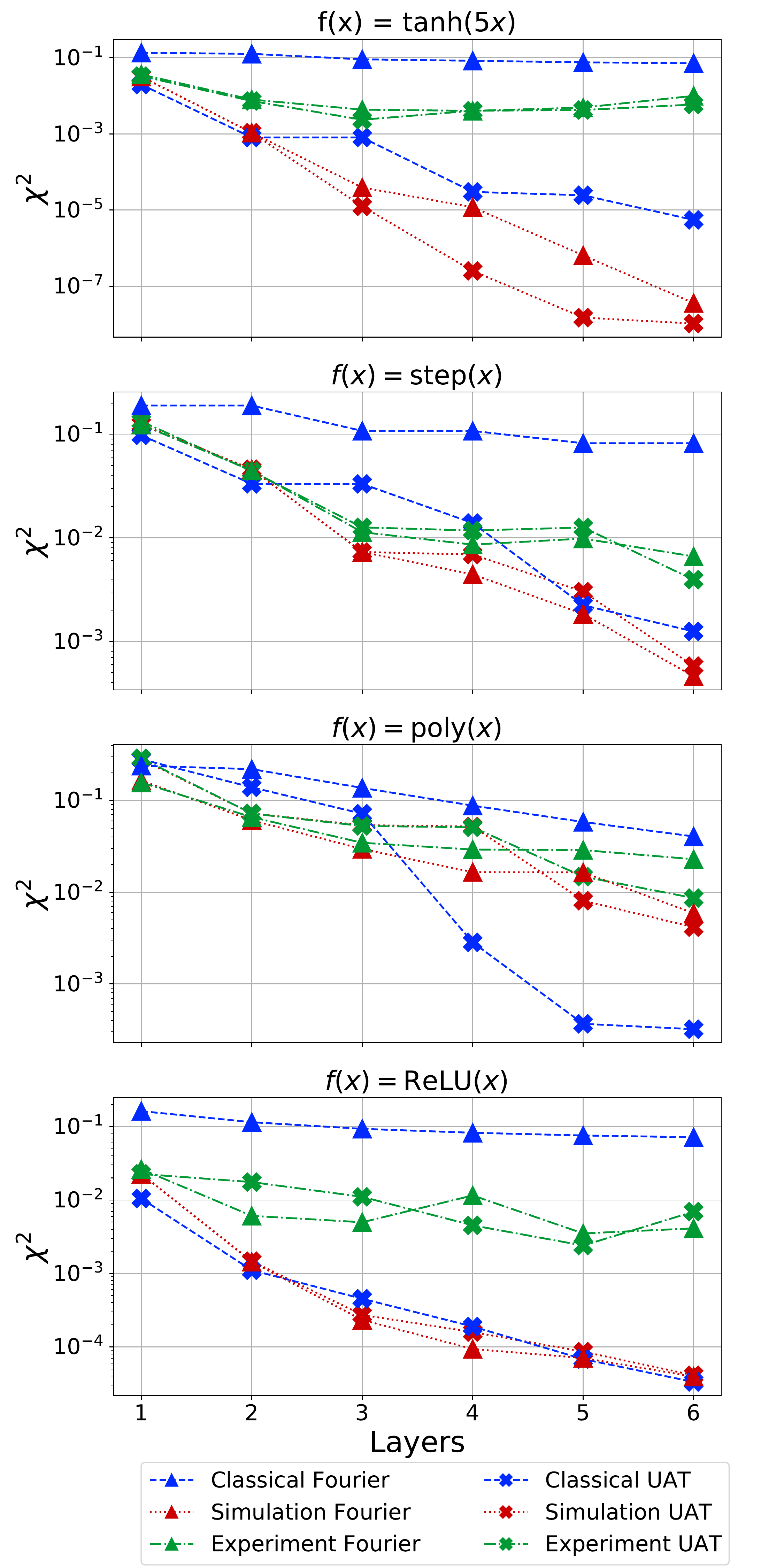}
    \caption{Values of $\chi^2$ for the $Z$ benchmark in all four test functions using classical computation (blue scatter), classical simulation of the quantum algorithm (red scatter) and experimental implementation with a superconducting qubit (green scatter). Fourier models are depicted with triangles, while UAT models are represented by crosses.}
    \label{fig:chi2_real}
\end{figure}

\begin{figure}[t!]
    \includegraphics[width=.8\linewidth]{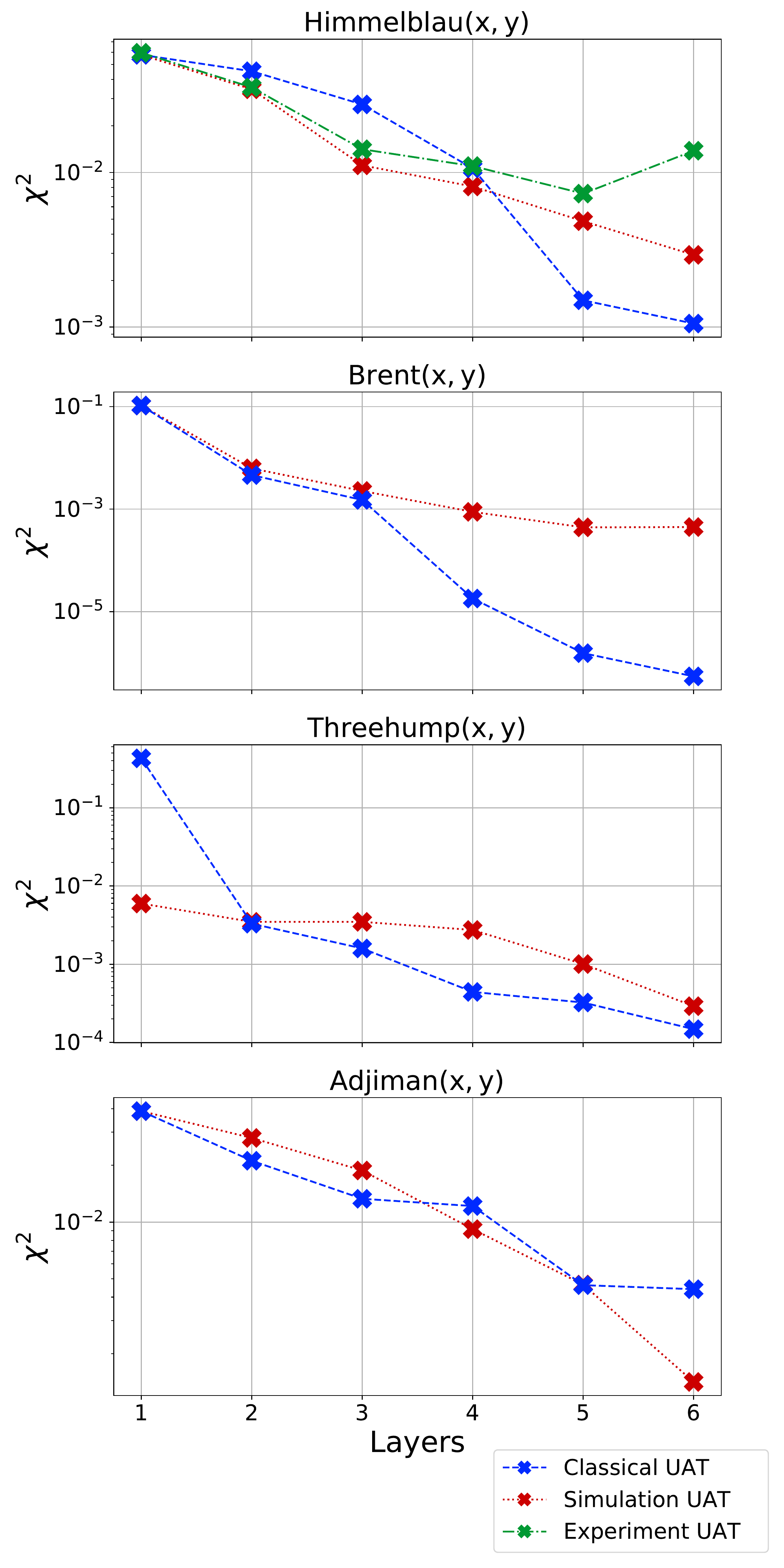}
    \caption{Values of $\chi^2$ for the $Z$ benchmark in all four test 2D functions using classical computation (blue scatter), classical simulation of the quantum algorithm (red scatter) and experimental implementation with a superconducting qubit (green scatter). Only UAT models are considered.}
    \label{fig:2d_chi}
\end{figure}

In all results presented in this section we provide three different final values. First, we use the Fourier and the UAT classical methods to approximate a target function. The Fourier method is obtained following the constructive recipe of Th.~\ref{th:fourier}. The UAT is applied using a single-hidden-layer Neural Network. Second, we approximate the same function using the quantum procedures defined in this work, simulating the wave function evolution with classical methods. In both cases, we retain the best outcome obtained with different initial conditions used in the optimization step. Finally, we use the parameters obtained using the simulation of the quantum procedure to execute that circuit in the actual superconducting quantum device. A specific set of $x$ values for the $\relu$ function from Eq.~\eqref{eq:functions} can be found in App.~\ref{app:experiment}. The theoretical optimal parameters may be, in principle, different than the experimental ones. Hence, an optimization performed directly on the experimental parameters could improve the final results \cite{future-work}.

We show in Fig.~\ref{fig:real_funs} the resulting fit for all four single-variable real-valued functions from Eqs.~\eqref{eq:functions}--\eqref{eq:functions3}. In this case the $Z$ benchmark with 5 layers is considered. A classical approximation (blue), a quantum exact simulation (red) and its experimental implementation (green) are depicted. 
All methods follow the overall shape of the target function. Classical Fourier approximations return less accurate predictions on the value of $f(x)$ due to the periodic nature of the model. The quantum Fourier and both classical and quantum UAT models return better results for all values of $x$. This behaviour is observed in all benchmarks. The experimental results retain the qualitative properties of the exact models, although a loss in performance is visible. In addition, an analysis of experimental uncertainties is also depicted at the UAT $\relu$ plot from Fig.~\ref{fig:real_funs}.

Figure~\ref{fig:2d_himmelblau} depicts the approximations obtained for the ${\rm Himmelblau (x, y)}$ function comparing the target function and all different methods considered. Figure~\ref{fig:2d_chi} summarizes the values of $\chi^2$ for all 2D-functions taken into account in this work. 

All different executions capture the overall shape of the function, but some differences exist in the different plots. Classical simulations return values for $Z < -1, Z > 1$, and thus lead to three minima in this case. On the other hand, the quantum simulation cannot clearly distinguish those minima. The experimental execution presents sharp contours because of the inherent noise and sampling uncertainty.

Figure~\ref{fig:chi2_real} shows a summary of the values of $\chi^2$ for classical and their analogous quantum simulated models and their experimental validation. In the case of classical and simulated quantum models a general trend towards better approximations --implying lower values of $\chi^2$-- is observed with an increasing numbers of layers.  

The simulated Fourier model performs better than its classical counterpart. This is due to the fact that a classical Fourier series does not contain tunable parameters, while its quantum version does.
However, the result from the classical Fourier series constitutes a lower bound for any approximation method based on optimization since at least the quality of the Fourier series is guaranteed. 

In the UAT case of Fig.~\ref{fig:chi2_real}, no approach returns better results. The classical algorithm performs better in the $\poly(x)$ case, but the results with the simulated quantum method improve the classical ones in the $\tanh(5x)$ case. Both models present similar trends as the number of layers increases. 

Despite the fact that the Fourier model contains more parameters than the UAT model, the latter performs better as seen in Fig.~\ref{fig:chi2_real}. Therefore, the UAT method seems more appropriate for the functions used here. 

The experimental realization of the quantum approximation models suffers from circuit noise and sampling uncertainties, and therefore degrades the quantity $\chi^2$. This is more prominent as more layers are added to the model. As a direct consequence, the approximation of the quantum model to the target function loses accuracy. The inherent sampling uncertainty sets a lower bound in the value of $\chi^2$ obtained through experiments.

In general, Fig.~\ref{fig:chi2_real} supports the claim that every layer grants the model more flexibility, and thus enhances the capability of fitting the target function. This flexibility is given by the number of re-uploadings of the independent variable and not by the amount of parameters. In addition, having too many parameters likely hinders the optimization procedure.

\begin{figure*}
    \centering
    \includegraphics[width=.6\linewidth]{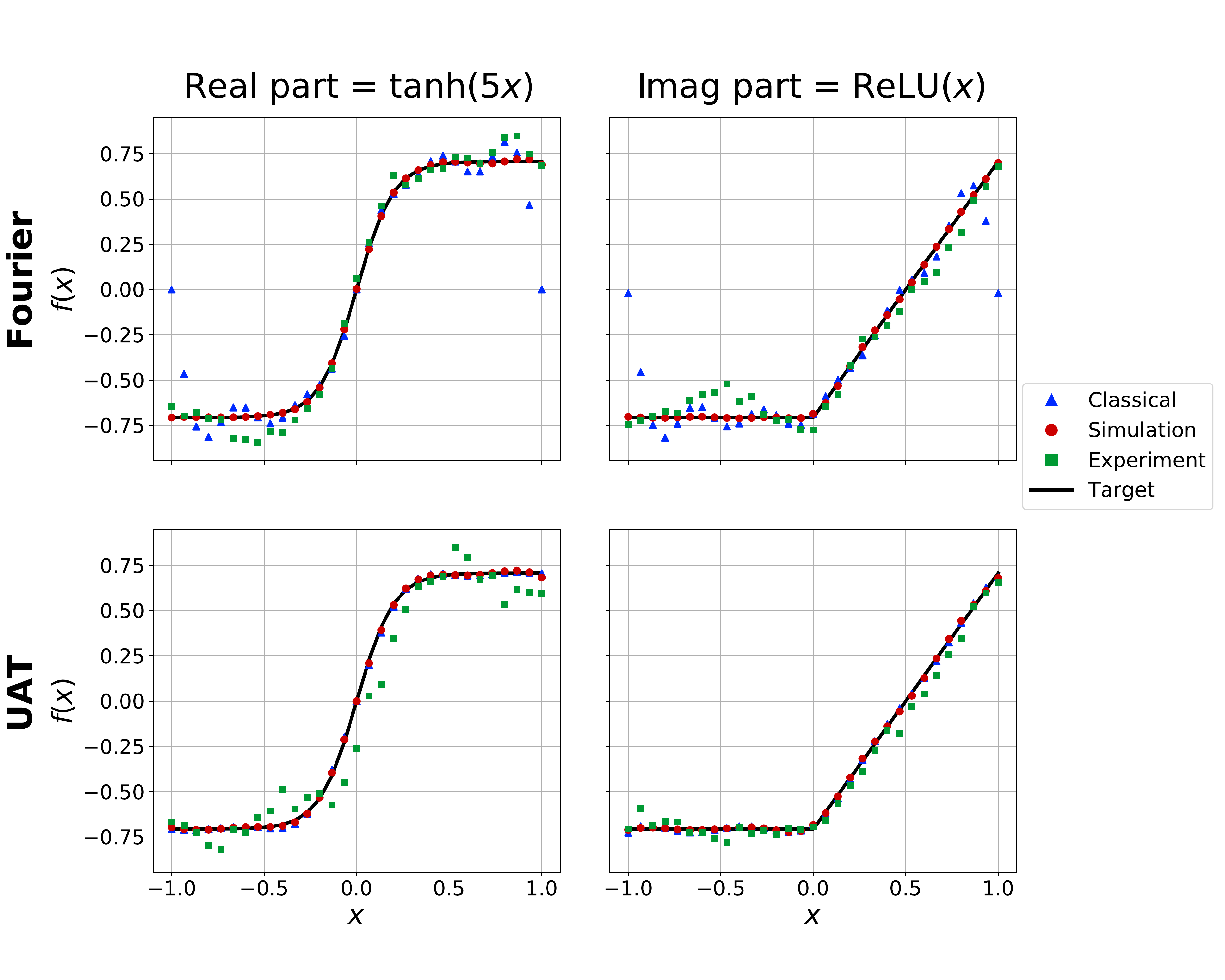}
    \caption{Fittings for the complex function $f(x) = \tanh(5x) + i \relu(x)$ properly normalized using the $X-Y$ benchmark for five layers. Blue triangles represent a classical model, while red dots represent its quantum counterparts computed using a classical simulator. Green squares are the experimental execution of the optimized quantum model using a superconducting qubit. The target function is plotted in black for comparison.}
    \label{fig:complex_funs}
    \vskip5mm
    \includegraphics[width=.85\textwidth]{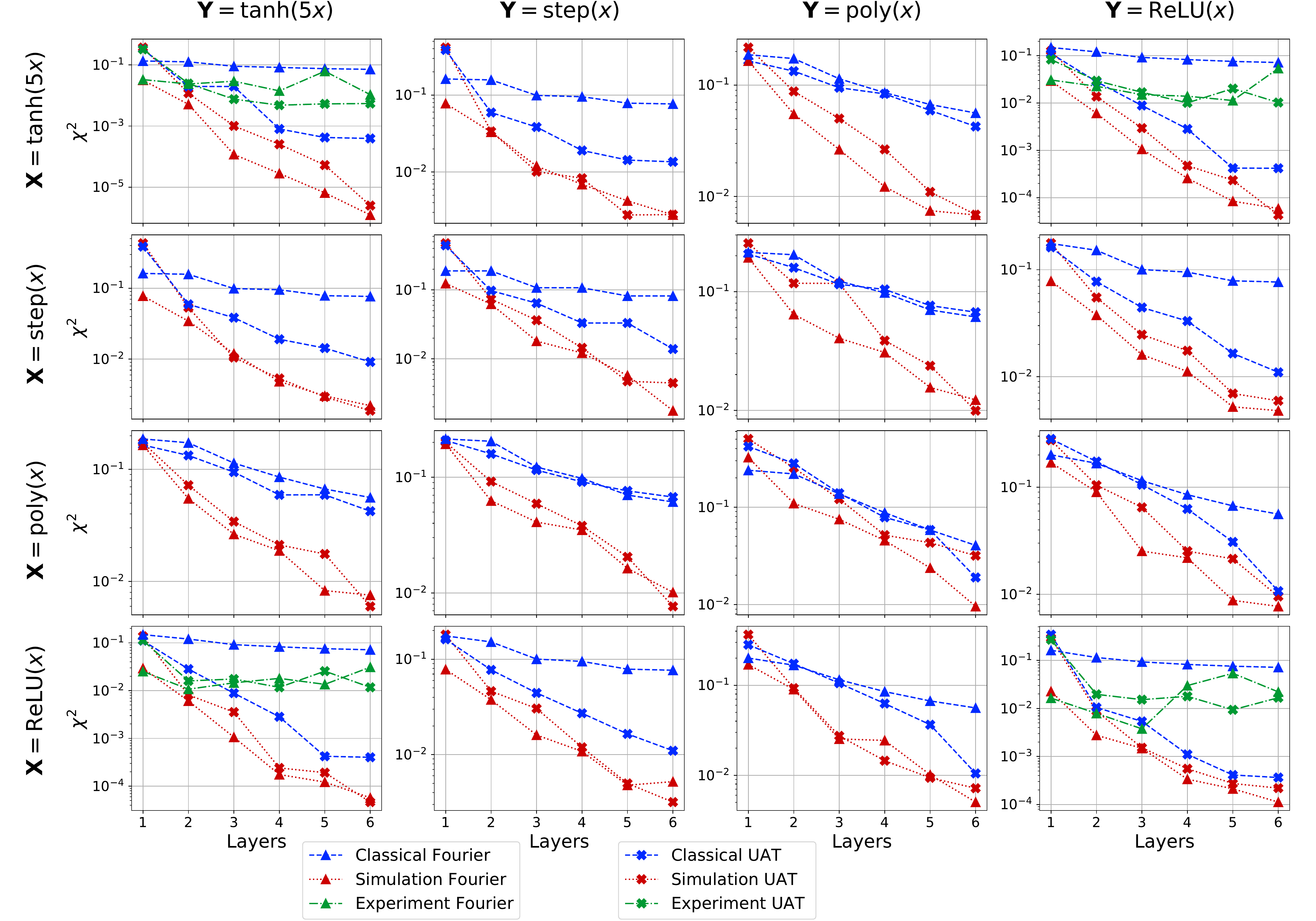}
    \caption{Values of $\chi^2$ for the X-Y benchmark in all possible combinations for real and imaginary parts of the four test functions from Eqs.~\eqref{eq:functions} to~\eqref{eq:functions3} using classical computation (blue scatter), classical simulation of the quantum algorithm (red scatter), and experimental implementation with a superconducting qubit (green scatter). Fourier models are depicted with triangles, while UAT models are represented by crosses.}
    \label{fig:chi2_complex}
\end{figure*}

The values of $\chi^2$ in Fig.~\ref{fig:2d_chi} measure the accuracy of the approximations. As before, we see that a larger number of layers provides better approximations to the target function. In agreement to the one-dimensional $Z$ benchmark, the scaling is similar for both quantum and classical methods.

A complex function in the $X-Y$ benchmark is depicted in Fig.~\ref{fig:complex_funs}. In that case, the $X$ measurement leads to $\tanh(5x)$ while the imaginary part contains $\relu(x)$. All the observations made for the $Z$ benchmark hold in this case. 

Fig.~\ref{fig:chi2_complex} shows values of $\chi^2$ for all possible combinations of real and imaginary parts using the functions described in Eqs.\eqref{eq:functions}--\eqref{eq:functions3}, being the real and imaginary parts. In this case, it is possible to see a common advantage for the quantum models. In particular, the functions $\tanh(5x)$ and $\relu(x)$ work better in any combination. This reflects the behaviour already observed in Fig.~\ref{fig:chi2_real}, where these functions present better performance than the other functions considered.

\section{Conclusions}\label{sec:conclusions}

We have shown that a single-qubit circuit has enough flexibility to 
encode any complex function $z(x)$ in the degrees of freedom of each quantum gate. 
This universal representation is achieved by acting with
a quantum circuit on a single-qubit 
gate that depends on input variables as well as additional parameters that are fixed by Machine Learning techniques.

This result guarantees that a single-qubit circuit, as defined in this work, is able to store two different and independent real functions. These functions are not restricted to be single-variable, as there exists no limitation to the dimensionality of its independent variable. Our present results provide the highest degree of compression of data in a single-qubit state, since there are no more degrees of freedom available in a qubit. 

The proof for universality was shown following two
different approaches, leading to two sets of single-qubit gates. 
In the first method, we found a link between quantum circuits and Fourier series. 
We have defined a quantum gate tuned by five parameters such that the 
row of $N$ gates applied to an initial state provides a final state 
where an $N$-term Fourier series is encoded. 
For the second method, a single-qubit quantum gate
is applied multiple 
times
to yield a final state whose form is 
compatible with the Universal Approximation Theorem. The input state does not compromise the validity of the approximation theorems but it affects the parameters defining the circuit. 

We also provide numerical evidence on the flexibility and approximation capabilities
of these quantum circuits. The benchmarks have been obtained using simulations and
classical minimizers to find optimal parameters for a set of test
functions. We have included as benchmarks 1D and 2D real functions and 1D complex
functions. The final results have also been compared to its classical
counterparts. In all cases, it is possible to see an equivalent scaling
for both classical and quantum methods. This ensures numerically that
the quantum procedure is comparable to the standard classical ones.

Experimental results implemented using a superconducting trasmon qubit confirm the same
trend obtained with the classical simulations. The finite qubit coherence
does not seem to impact the results significantly in the gate sets applied.

It is known that no single-qubit algorithm can bring quantum advantage since it is efficiently simulable by classical computers. The approach here presented can instead be included as a sub-task of a larger routine. Achieving a quantum advantage requires extending the approach to multi-qubit circuits. For instance, one can implement the gate set defined in this work to apply on several different qubits, and then add entangling gates, see \cite{reuploading-perezsalinas2020, proton-perezsalinas2020}. Whether this extension to multi-qubit circuits is more flexible than the approach of this work remains unclear. Thus, the present work can serve as a starting point for studying the representation
capability of quantum systems beyond one qubit
(see also \cite{encoding-schuld2020}).

\section*{Acknowledgements}
We thank Martin Weides and Marco Pfirrman at Glasgow University for fabricating the superconducting transmon qubit device used in this work at the Karlsruhe Institute of Technology (KIT), and Prof. Sergio O. Valenzuela from the Catalan Institute of Nanoscience and Nanotechnology (ICN2) for granting access to the dilution refrigerator during the initial stages of the measurements. We also thank the rest of the IFAE QCT group and the Qilimanjaro team for their contribution setting up the new laboratory space where part of the experiment was conducted. We acknowledge financial support from Secretaria d'Universitats i Recerca del Departament d'Empresa i Coneixement de la Generalitat de Catalunya, co-funded by the European Union Regional Development Fund within the ERDF Operational Program of Catalunya (project QuantumCat, ref. 001-P-001644). A.G-S received funding from the European Union’s Horizon 2020 research and innovation programme under grant agreement No 951911 (AI4Media). P. F.-D. acknowledges support from "la Caixa" Foundation - Junior leader fellowship (ID100010434-LCF/BQ/PR19/11700009), Ministry of Economy and Competitiveness and Agencia Estatal de Investigación (FIS2017-89860-P; SEV-2016-0588; PCI2019-111838-2), and European Commission (FET-Open AVaQus GA 899561; QuantERA). IFAE is partially funded by the CERCA program of the Generalitat de Catalunya.

\bibliography{citations.bib}

\begin{thebibliography}{41}
\bibitem [{fut()}]{future-work}%
  \BibitemOpen
  \bibinfo {title} {A direct optimization of the experimental parameters from
  the universal approximant using the qubit is currently being carried out and
  will be reported in future work.}\BibitemShut {Stop}%
\end{thebibliography}

\newpage
\clearpage
\newpage
\appendix

\section{Proof of universality theorems}
\noindent We prove here the results claimed in Theorems \ref{th:q_fourier} and \ref{th:q_UAT}.

\subsection{Demonstration for the quantum Fourier series}\label{app:fourier}

The quantum circuit proposed in Theorem \ref{th:q_fourier} fulfills the requirement that every new gate plays the role of a new step in the original Fourier series. The proof is based on an inductive procedure  and can be then decomposed in two steps. First, we show that the first gate of the circuit is equivalent to the $0$-th constant Fourier term. Then, we show that if there are $N$ gates in a row forming a $N$-term Fourier series, then adding a new gate provides a $(N+1)$-terms Fourier series.

\bigskip

Let the fundamental gate $U^\mathcal{F}(x, \vec\theta)$ gate defined in Eq. \eqref{eq:unitary_f} be
\begin{widetext}
\begin{multline}
 U^\F(x; \vec\theta) = U^\mathcal{F}(x; \omega, \alpha, \beta, \varphi, \lambda) = R_z\left(\alpha + \beta\right) R_y(2\lambda)R_z\left(\alpha - \beta\right) R_z(2\omega x) R_y(2\varphi) = \\ = 
    \begin{pmatrix}
    \cos\lambda \cos\varphi e^{i \alpha} e^{i\omega x} - \sin\lambda \sin\varphi e^{i \beta} e^{- i \omega x} & 
    -\cos\lambda \sin\varphi e^{i \alpha} e^{i\omega x} - \sin\lambda \cos\varphi e^{i \beta} e^{- i \omega x} \\
    \sin\lambda \cos\varphi e^{-i \beta} e^{i\omega x} + \cos\lambda \sin\varphi e^{-i \alpha} e^{- i \omega x} & 
    -\sin\lambda \sin\varphi e^{-i \beta} e^{i\omega x} + \cos\lambda \cos\varphi e^{-i \alpha} e^{- i \omega x} \\
    \end{pmatrix},
\end{multline}
\end{widetext}

It is possible to recast the above choice of fundamental gate using the following redefinition of parameters,
\begin{eqnarray}
    a_+  =   \cos\lambda \cos\varphi e^{i \alpha}, \\
    a_-   =   -\sin\lambda \sin\varphi e^{i \beta}, \\
    b_+  =   -\cos\lambda \sin\varphi e^{i \alpha}, \\
    b_-  =  - \sin\lambda \cos\varphi e^{i \beta}.
\end{eqnarray}

A more compact representation of the fundamental gate follows

\bigskip
\begin{lemma}

The fundamental gate can be expressed as
\begin{multline}\label{eq:u_f_2}
\small
    U^{\mathcal{F}}(x; \omega, \alpha, \beta, \varphi, \lambda) = \\
    \begin{pmatrix}
    a_+ e^{i\omega x} + a_- e^{- i \omega x} & 
    b_+ e^{i\omega x} + b_- e^{- i \omega x} \\
    -b_-^{*} e^{i\omega x} - b_+^{*} e^{- i \omega x} & 
    a_-^{*} e^{i\omega x} + a_+^{*} e^{- i \omega x} \\
    \end{pmatrix},
\end{multline}

as can be verified by simple substitution from Definition \ref{def:fourier_gate}. 

\end{lemma} 
Note that this expression corresponds to a unitary matrix, due to the relations involved in the definition of the coefficients $a_\pm$ and $b_\pm$.
Note also that a unitary matrix has three degrees of freedom, which are here fixed by 5 parameters. An intuition behind the role of these parameters is that $\alpha, \beta, \varphi, \lambda$ are related to the coefficients of one Fourier step, that is $a_\pm, b_\pm$, while $\omega$ can be identified with the corresponding frequency.

A total circuit can be constructed by multiplying $k$ fundamental gates to obtain $\mathcal{U}^{(k), s}_{f, \phi}$ as in Definition \ref{def:prod_general}. Starting with this composite gate, we can now proof the main Fourier approximation theorem.
\bigskip

\begin{theorem}\label{th:fourier_2}
There exists a series of $k$ single-qubit gates forming a $k$-th approximant circuit that delivers a unitary operation where all its coefficients are written as Fourier series.
\end{theorem} 

\begin{proof}
The proof of this constructive theorem consists in making contact with harmonic analysis and proceeds by induction.

{\bf i)} The first circuit consists only of one fundamental gate, chosen with frequency $\omega=0$, that is
\begin{equation}
    U_0^\mathcal{F} = 
    \begin{pmatrix}
    A_0 &  B_0 \\
    -B_0^{*} & A_0^{*} 
    \end{pmatrix},
\end{equation}
This, indeed corresponds to the first constant term of Fourier series.

{\bf ii)} We now assume that the $N$-th approximant circuit takes the form
\begin{equation}
\small
    \prod_{i=0}^N U_i^\mathcal{F} = \begin{pmatrix}
    \sum_{n = -N}^N A_n e^{i \Omega_n x} & \sum_{n = -N}^N B_n e^{i \Omega_n  x} \\
    -\sum_{n = -N}^N B_n^ {*} e^{-i \Omega_n  x} & \sum_{n = -N}^N A_n^ {*} e^{-i \Omega_n  x}
    \end{pmatrix}.
\end{equation}
where the frequencies are  $(\Omega_n \pm \omega)$. The  result of adding a new fundamental gate corresponds to 
\begin{equation}\label{eq:uN+1}
\small
\prod_{i=0}^{N+1} U_i^\mathcal{F} =
\begin{pmatrix}
    \sum_{n = -N-1}^{N + 1} \tilde{A}_n e^{i \tilde \Omega_n x} & \sum_{n = -N-1}^{N + 1} \tilde{B}_n e^{i \tilde \Omega_n x} \\
    -\sum_{n = -N-1}^{N + 1} \tilde{B}_n^{*} e^{-i \tilde \Omega_n x} & \sum_{n = -N-1}^{N + 1} \tilde{A}_n^{*} e^{-i \tilde \Omega_n x}
    \end{pmatrix}
\end{equation}{}
where we need to fix the new coefficients $\tilde\Omega_n$ and frequencies in terms of the old ones $\Omega_n$ and the new single gate frequency $\omega$ added to the circuit. It is easy to see that the addition of a gate changes the frequency in one unit, that is, $\tilde\Omega=\Omega_n \pm \omega$. Then, the general structure of the series can be adapted to  a Fourier expansion by choosing
\begin{equation}
    \Omega_n=(2 n+1)\frac{\pi}{2}.
\end{equation}

After fixing the values that the frequencies must take, it is straightforward to re-arrange terms in the matrix and reach
\begin{eqnarray}\label{eq:induction}
\tilde{A}_0 & = & A_0 a_- - B_0^{*} b_- \\
\tilde{A}_{\pm n} & = & A_{\pm n} a_- - B_{\mp n}^{*} b_- \nonumber \\
 & + & A_{\pm (n-1)} a_+ - B^{*}_{\mp (n-1)} b_+ \\
\tilde{A}_{\pm (N + 1)} & = & A_{\pm N} a_+ - B^{*}_{\mp N} b_+\\
\tilde{B}_0 & = & B_0 a_- + A^{*}_0 b_- \\ 
\tilde{B}_{\pm n} & = & B_{\pm n} a_- + A^{*}_{\mp n} b_- \nonumber \\
 & + & B_{\pm (n-1)} a_+ + A^{*}_{\mp (n-1)} b_+ \\
\tilde{B}_{\pm (N + 1)} & = & A^{*}_{\mp N} a_+ + A^{*}_{\mp N} b_+
\end{eqnarray}
This provides the explicit connection between approximant circuits and Fourier expansions for the coefficients of the global unitary matrix.

\end{proof}

The above constructive theorem is sufficient to prove that the output probability of a series of approximant circuits
can reproduce any functionality.

\bigskip

\subsection{Demonstration for the quantum UAT}\label{app:uat}

An alternative manner to design a single-qubit universal approximant is related to the equivalent Universal Approximation Theorem broadly used in Neural Networks  \cite{uat-cybenko1989}. The idea is to start from a different fundamental gate.

\bigskip

Let the {\sl fundamental gate} $U^{\rm UAT}(\vec x; \vec\theta)$ defined in Eq. \eqref{eq:unitary_uat} be explicitly
\begin{widetext}{}
\begin{equation}\label{eq:u_uat}
    U^{\rm UAT}(x; \vec \omega, \alpha, \varphi) = R_z(2\left(\vec \omega \cdot \vec x + \alpha\right)) R_y(2\varphi) = 
    \begin{pmatrix}
        \cos(\varphi) e^{i(\vec \omega \cdot \vec x + \alpha)} & -\sin(\varphi) e^{i(\vec \omega \cdot \vec x + \alpha)} \\
        \sin(\varphi) e^{-i(\vec \omega \cdot \vec x + \alpha)} & \cos(\varphi) e^{-i(\vec \omega \cdot \vec x + \alpha)}
        \end{pmatrix},
\end{equation}
\end{widetext}

A total circuit can be constructed by multiplying $k$ fundamental gates to obtain $\mathcal{U}^{(k), {\rm UAT}}_{f, \phi}$ as in Definition \ref{def:prod_general}. We can now prove the quantum UAT using this fundamental gate. 

\begin{theorem}\label{th:uat_2}
There exists a series of $k$ single-qubit gates forming a $k$-th approximant circuit that delivers a unitary operation where all its coefficients are written as an approximation as defined by Theorem \ref{th:UAT}, UAT.
\end{theorem}
\begin{proof}
Let us take the $U^{\rm UAT}$ defined in Eq. \eqref{eq:u_uat}.
\begin{multline*}
    R_z(2 \vec \omega \cdot \vec x + 2\alpha) R_y(2\varphi) = \\ \begin{pmatrix}
\cos(\varphi) e^{i(\vec \omega \cdot \vec x + \alpha)} & -\sin(\varphi) e^{i(\vec \omega \cdot \vec x + \alpha)} \\
\sin(\varphi) e^{-i(\vec \omega \cdot \vec x + \alpha)} & \cos(\varphi) e^{-i(\vec \omega \cdot \vec x + \alpha)}
\end{pmatrix}
\end{multline*}

By direct inspection it is straigthforward to check that every entry in this matrix can be understood as one term of $\bar f_N$ in Eq. \eqref{eq:UAT}. From this definition we can infer the recursive rule that defines all steps. If 
\begin{eqnarray}
A_{N} = \bra 0 \prod_{n=1}^{N}U_n^{UAT} \ket 0  \\
B_{N} = \bra 1 \prod_{n=1}^{N}U_n^{UAT} \ket 0
\end{eqnarray}

then the updating rule is 

\begin{eqnarray}
\begin{split}
A_{N+1} = 
A_N  \cos(\varphi_{N+1}) e^{i \vec \omega_{N+1} \cdot \vec x} e^{i \alpha_{N+1}} - \\
B_N  \sin(\varphi_{N+1}) e^{i \vec \omega_{N+1} \cdot \vec x} e^{i \alpha_{N+1}}    
\end{split}
 \\
 \begin{split}
B_{N+1}=
A_N \sin(\varphi_{N+1}) e^{-i \vec \omega_{N+1} \cdot \vec x}e^{\alpha_{N+1}} + \\
B_N \cos(\varphi_{N+1}) e^{-i \vec \omega_{N+1} \cdot \vec x}e^{i\alpha_{N+1}}     
 \end{split}
\end{eqnarray}

Having this updating rule in mind, it is possible to write 
\begin{equation}\label{eq:B_N}
B_N = \sum_{m = 0}^{2^{N - 1}} c_m(\varphi_1, \ldots, \varphi_N) e^{i \delta_m(\alpha_1, \ldots, \alpha_N)} e^{i \vec w_m(\vec\omega_1, \ldots, \vec\omega_N) \cdot \vec x}, 
\end{equation}
where the inner dependencies of $c_m$ are products of sines and cosines of $\varphi_n$, and those of $\delta_m$ and $\vec w_m$ are linear combinations of $\alpha_n$ and $\omega_n$.

Let us proceed now as in the proof of the UAT in Ref. \cite{uat-cybenko1989}. Let us take $S$ as the set of functions of the form $B_N(\vec x)$, and $C^{\mathbb{C}}(I_m)$ the set of continuous complex-valued functions in $I_m$, defined as in Theorem \ref{th:UAT}. We assume that $S \subset C^{\mathbb{C}}(I_m)$, and $S \neq C^{\mathbb{C}}(I_n)$. We can now apply the Theorem \ref{th:hahn_banach}, known as Hahn-Banach theorem. This theorem allows to state that there exists a linear functional $L$ acting on $C^{\mathbb{C}}(I_n)$ such that
\begin{equation}
L(S) = L(\bar{S}) = 0, \qquad L\neq 0.
\end{equation}
Notice that this theorem is applicable as there are no restriction in working only with real numbers.

We call now Theorem \ref{th:riesz}, known as Riesz representation theorem. We can write the functional $L$ as 
\begin{equation}
L(h) = \int_{I_n} h(x) d\mu(x)
\end{equation}
for $\mu \in M(I_n)$ non-null and $\forall \, h \in  C^{\mathbb{C}}(I_n)$. In particular, 
\begin{equation}
L(h) = A_N(\vec x) d\mu(\vec x) = 0,
\end{equation}
and thus
\begin{equation}
\int_{I_n} e^{i\vec{v_m}(\omega_1, \ldots, \omega_N) \cdot \vec x} d\mu(\vec x) = 0.
\end{equation}
This is the usual Fourier transform of $\mu$. We can conclude by calling Theorem \ref{th:lebesgue}, Lebesgue Bounded Convergence theorem, that if the $\mathcal{FT}(\mu) = 0$, then $\mu = 0$, and we come into a contradiction with the only assumption we made. 

The measure of all half-planes being 0 implies that $\mu = 0$. Let us fix $\vec w$, and for a bounded measurabe function $h$ we define the linear functional
\begin{equation}
F(h) = \int_{I_n} h(\vec w \cdot \vec x) d\mu(x),
\end{equation}
which is bounded on $L^\infty(\mathbb{R})$ since $\mu$ is a finite signed measure. Let $h$ be an indicator of the half planes $h(u) = 1$ if $u\geq -b$ and $h(u) = 0$ otherwise, then
\begin{equation}
F(h) = \int_{I_n} h(\vec w \cdot \vec x) d\mu(x) =  \mu(\Pi_{\vec w, b}) + \mu(H_{\vec w, b}) = 0.
\end{equation}
By linearity, $F(h) = 0$ for any simple function, such as sum of indicator functions of intervales \cite{analysis-ash1972}. 

In particular, for the bounded measurable functions $s(u) = \sin(\vec w \cdot \vec x), c(u) = \cos(\vec w \cdot \vec x)$ we can write
\begin{equation}
F(c + is) = \int_{I_n} \exp{i \vec w \cdot \vec x} d\mu(\vec x) = 0.
\end{equation}
The Fourier Transform of this $F$ is null, thus $\mu = 0$.

\end{proof}

For the sake of completeness, we cover now the three theorems required for the proof.

\begin{theorem}\label{th:hahn_banach}
{\bf: Hahn-Banach} \cite{analysis-hahn1927, analysis-banach1929}\\

Set $\mathbb{K} = \mathbb{R} {\;\rm or\;} \mathbb{C}$. Let $V$ be a $\mathbb{K}-$vector space with a seminorm $p: V \rightarrow \mathbb{R}$. If $\varphi : U \rightarrow \mathbb{K}$ is a $\mathbb{K}-$linear functional on a $\mathbb{K}-$linear subspace $U\subset V$ such that
\begin{equation}
|\varphi(x)| \leq p(x) \qquad \forall x \in U,
\end{equation}
then there exists a linear extension $\psi : V \rightarrow \mathbb{K}$ of $\varphi$ to the whole space $V$ such that
\begin{eqnarray}
\psi(x) = \varphi(x) \qquad \forall x\in U \\
|\psi(x)| \leq p(x) \qquad \forall x\in V
\end{eqnarray}
\end{theorem}

\begin{theorem}\label{th:riesz}
{\bf: Riesz Representation} \cite{analysis-riesz1914}\\

Let $X$ be a locally compact Hausdorff space. For any positive linear functional $\psi$ on $C(X)$, there exists a uniruq regular Borel measure $\mu$ such that
\begin{equation}
\forall f \in C_c(X): \qquad \psi(f) = \int_X f(x) d\mu(x)
\end{equation}
\end{theorem}

\begin{theorem}\label{th:lebesgue}
{\bf: Lebesgue Bounded Convergence} \cite{analysis-weir1974}\\
Let $\lbrace f_n\rbrace$ be a sequence of complex-valued measurable functions on a measure space $(S, \Sigma, \mu)$. Suppose that $\lbrace f_n \rbrace$ converges pointwise to a function $f$ and is dominated by some integrable function $g(x)$ in the sense
\begin{equation}
|f_n(x)| \leq g(x), \qquad \int_S |g|d\mu < \infty
\end{equation}
then
\begin{equation}
\lim_{n\rightarrow \infty} \int_S f_n d\mu = \int_S f d\mu
\end{equation}
\end{theorem}

\subsection{Link to output of quantum circuits}
Last sections were devoted to prove that specific series of circuits return functionalities able to represent a wide range of functions. In this last step we relate previous results to the output of quantum circuits.

\begin{theorem}
The computational basis output of a single-qubit quantum circuit can provide a convergent approximattion to any desired function.
\end{theorem}  

\begin{proof}
The output of a $k$-th approximant circuit can be cast a an approximation expansion of an arbitrary function.
It is sufficient to initalize a register in the $\vert 0\rangle$ state and measure the output in the computational basis. It follows 
\begin{equation}
    \bra{1} \prod_{i=0}^{N} U_i^s \ket{0} =  z_N(x)
\end{equation}
where $z_N(x)$ can take different forms. 

If the fundamental gate is $U^\mathcal{F}$, then the output is the truncated Fourier series
\begin{equation}
    z_N(x) = \sum_{n=-N}^N B_n e^{i 2\pi n x},
\end{equation}
where $B_n$ are free complex coefficients. This result holds for single-variable functions.

If the fundamental gate is $U^{\rm UAT}$, then the output is a function
\begin{equation}
\small
    z_N(\vec x) = \sum_{m = 0}^{2^{N - 1}} c_m(\varphi_1, \ldots, \varphi_N) e^{i \delta_m(\alpha_1, \ldots, \alpha_N)} e^{i \vec w_m(\vec\omega_1, \ldots, \vec\omega_N) \cdot \vec x}, 
\end{equation}
according to Eq. \eqref{eq:B_N}. This result holds for single- and multi-variable functions.

According to theorems \ref{th:fourier_2} and \ref{th:uat_2}, both expressions can approximate any desired function.

\end{proof}

\bigskip

\section{UAT for complex functions}\label{app:real_to_complex_uat}
In this Appendix we show that the standard formulation of the UAT supports the approximation of complex function using $e^{i (\cdot)}$ as the activation function.

Let us 
follow the approximations according to the UAT of the function
\begin{equation}
    z(\vec x) = a(\vec x) + i b(\vec x), 
\end{equation}
using trigonometric functions as $\sigma(\cdot)$, 
\begin{eqnarray}
a(x) = \sum_{j=1}^N \alpha_i \cos(\vec w_j \cdot \vec x + a_j) \\
b(x) = \sum_{j=1}^N \beta_i \sin(\vec v_j \cdot \vec x + b_j).
\end{eqnarray}
Then
\begin{multline}
z(x) = \sum_{j=1}^N \alpha_i \cos(\vec w_j \cdot \vec x + a_j) + \\ i \sum_{j=1}^N \beta_i \sin(\vec v_j \cdot \vec x + b_j),
\end{multline}

and this equation is can be rearranged as 
\begin{multline}
    z(x) = \sum_{j=1}^N \frac{\alpha_j}{2}\left( e^{i (\vec w_j \cdot \vec x + a_j)} + e^{-i ( \vec w_j \cdot \vec x + a_j)}\right) + \\ \frac{\beta_j}{2} \left(e^{i (\vec v_j \cdot \vec x + b_j)} - e^{-i (\vec v_j \cdot \vec x + b_j)} \right),
\end{multline}

what encourages the UAT formulation for complex functions as an analogous to Eq. \eqref{eq:UAT}
\begin{equation}\label{eq:complex_UAT}
    G(\vec x) = \sum_{n=1}^N \gamma_n e^{i\delta_n} e^{i \vec u_n \cdot \vec x}.
\end{equation}

\section{2D functions for benchmark}\label{app:2D_benchmark}

\begin{figure}
    \flushleft
    \includegraphics[width=.5\textwidth]{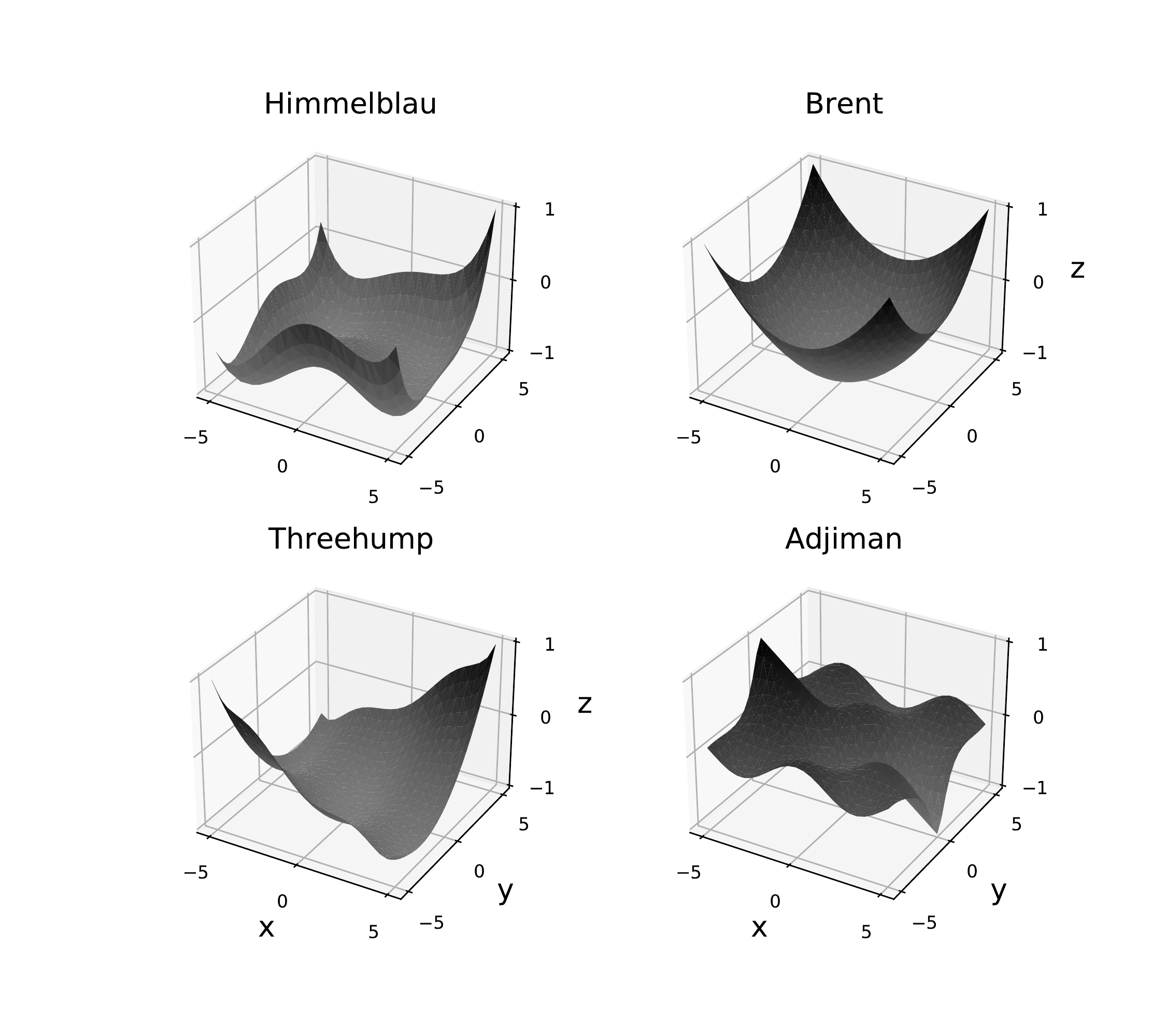}
    \caption{Graphical representation of 2-dimensional functions utilized for benchmarking. A regularization is applied to obtain $Z$ values between $-1$ and $1$ for the given domain.}
    \label{fig:2d_functions}
\end{figure}

The definitions used for the 2-dimensional functions \cite{2d_functions} that serve for benchmarking our proposed algorithms are define as

\begin{multline}
    {\rm Himmelblau}(x, y)  = \\ = (x^2 + y - 11)^2 + (x + y^2 - 7)^2,
\end{multline}

\begin{multline}
{\rm Brent}(x, y) = \\ =    \left(\frac{x}{2}\right)^2 + \left(\frac{y}{2}\right)^2 + e^{-\left(\left(\frac{x}{2} - 5\right)^2 + \left(\frac{y}{2} - 5\right)^2)\right)},
\end{multline}

\begin{multline}
{\rm Threehump}(x, y) = \\ = 2 \left(\frac{2 x}{5}\right)^2 - 1.05 \left(\frac{2 x}{5}\right)^4 + \frac{1}{6}\left(\frac{2 x}{5}\right)^6 + \\ + \left(\frac{2 x}{5}\right) \left(\frac{2 y}{5}\right) + \left(\frac{2 y}{5}\right)^2,
\end{multline}

\begin{multline}
{\rm Adjiman}(x, y) = \cos(x) \sin(y) - \frac{x}{y^2 + 1},
\end{multline}

where a normalization to $-1 \leq f(x, y) \leq 1$ is applied after this definition. A graphical representation of these functions is depicted in Fig. \ref{fig:2d_functions}.

\section{Experimental methods}\label{app:experiment}
\begin{table*}
    \centering
    \resizebox{.8\linewidth}{!}{\begin{tabular}{ *{13}{l} }
      \toprule
      \multicolumn{1}{c}{\textbf{Optimal}} 
      & \multicolumn{1}{c}{$p_1$}  & \multicolumn{1}{c}{$p_2$}  & \multicolumn{1}{c}{$p_3$}
      & \multicolumn{1}{c}{$p_4$}  & \multicolumn{1}{c}{$p_5$}  & \multicolumn{1}{c}{$p_6$}
      & \multicolumn{1}{c}{$p_7$}  & \multicolumn{1}{c}{$p_8$}  & \multicolumn{1}{c}{$p_9$}
      & \multicolumn{1}{c}{$p_{10}$}  & \multicolumn{1}{c}{$p_{11}$}  & \multicolumn{1}{c}{$p_{12}$}
      \\
      \textbf{parameters}
      & -2.501 & 1.685 &   1.757 & 2.105 & 3.822 & -1.788 &
 -1.507 & -4.640 & 0.430 & 1.875 & 5.038 & -1.906
      \\[0.2cm]
      \hline
      \multicolumn{1}{c}{\textbf{Rotational}} 
      & \multicolumn{2}{c}{$Z_1$}   & \multicolumn{1}{c}{$Y_1$}
      & \multicolumn{2}{c}{$Z_2$}   & \multicolumn{1}{c}{$Y_2$}
      & \multicolumn{2}{c}{$Z_3$}   & \multicolumn{1}{c}{$Y_3$}
      & \multicolumn{2}{c}{$Z_4$}   & \multicolumn{1}{c}{$Y_4$}
      \\
      \multicolumn{1}{c}{\textbf{angles$^*$}} 
      & \multicolumn{2}{c}{$p_1  + p_2 x$}   & \multicolumn{1}{c}{$p_3$}
      & \multicolumn{2}{c}{$p_4 + p_5 x$}   & \multicolumn{1}{c}{$p_6$}
      & \multicolumn{2}{c}{$p_7 + p_8 x$}   & \multicolumn{1}{c}{$p_3$}
      & \multicolumn{2}{c}{$p_{10} + p_{11} x$}   & \multicolumn{1}{c}{$p_{12}$}
      \\[0.2cm]
      \multicolumn{1}{c}{\textbf{$x=-0.5$}} 
      & \multicolumn{2}{c}{2.939}
& \multicolumn{1}{c}{1.757}
& \multicolumn{2}{c}{0.194}
& \multicolumn{1}{c}{4.495}
& \multicolumn{2}{c}{0.813}
& \multicolumn{1}{c}{0.430}
& \multicolumn{2}{c}{5.639}
& \multicolumn{1}{c}{4.377}
      \\
      \multicolumn{1}{c}{\textbf{$x=0$}} 
      & \multicolumn{2}{c}{3.782}
& \multicolumn{1}{c}{1.757}
& \multicolumn{2}{c}{2.105}
& \multicolumn{1}{c}{4.495}
& \multicolumn{2}{c}{4.776}
& \multicolumn{1}{c}{0.430}
& \multicolumn{2}{c}{1.875}
& \multicolumn{1}{c}{4.377}
      \\
      \multicolumn{1}{c}{\textbf{$x=1$}} 
      & \multicolumn{2}{c}{5.467}
& \multicolumn{1}{c}{1.757}
& \multicolumn{2}{c}{5.927}
& \multicolumn{1}{c}{4.495}
& \multicolumn{2}{c}{0.136}
& \multicolumn{1}{c}{0.430}
& \multicolumn{2}{c}{0.630}
& \multicolumn{1}{c}{4.377}
      \\[0.2cm]
      \hline
      \multicolumn{9}{l}{\small  $^*$ Angles between $0$ and $2\pi$}        
    \end{tabular}}
    \caption{Optimal parameters and angles obtained for $\relu(x)$ and 4 layers. Above the 12 parameters that define the rotational angles obtained through simulations. Below the corresponding angles of the 8 rotations for three different values of $x$. Note that Y rotations are not $x$-dependent, hence they are equal for all three $x$ values.}
    \label{tab:rotation angles}
    
\end{table*}
\begin{figure}[t!]
\centering
\includegraphics[width=\linewidth]{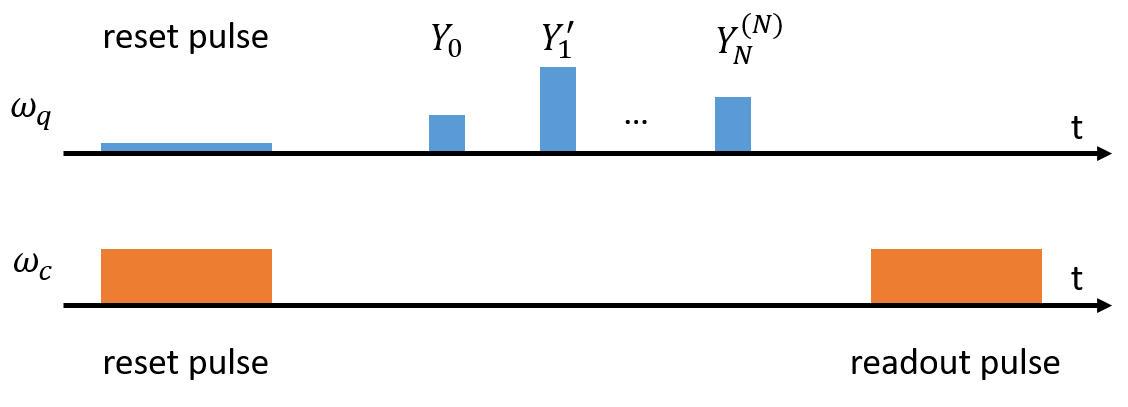}
\caption[Full Sequence]{Complete pulse sequence. First, the reset protocol is performed which corresponds to two pulses at the cavity and the qubit frequencies, respectively. Note that the qubit pulse is of considerably lower amplitude than the cavity pulse. Also, both pulses have a longer duration than the qubit rotation sequence (timings not to scale). The ``Y" pulses are shown to have different amplitudes to determine each rotation angle. Finally, the readout corresponds to a pulse at the cavity frequency which is later read out by a digitizing card.}
\label{fig:full-seq} 
\end{figure}

\begin{figure}
\centering
\includegraphics[width=\linewidth]{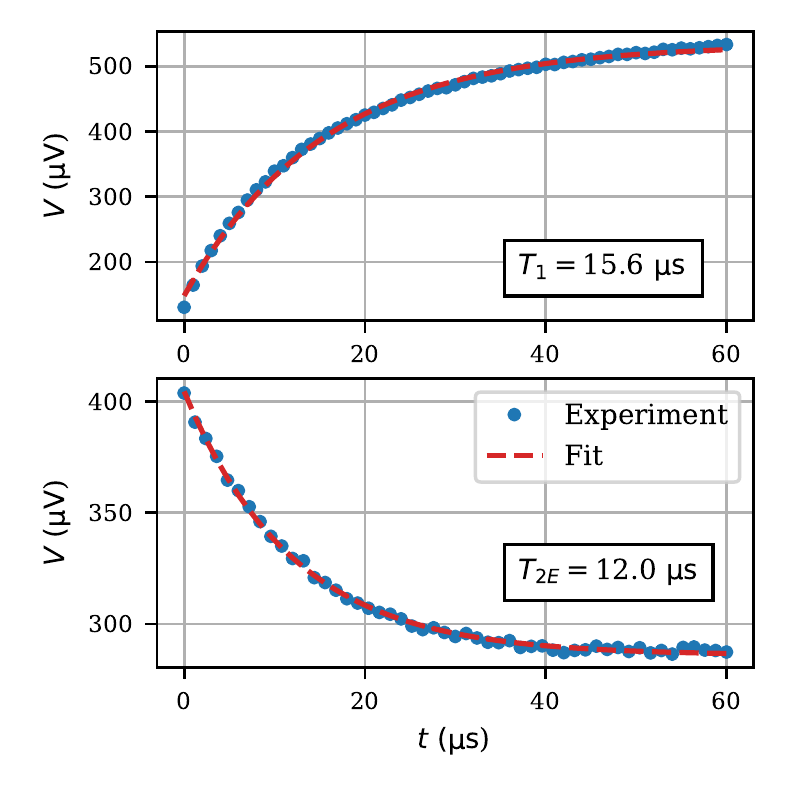}
\includegraphics[width=\linewidth]{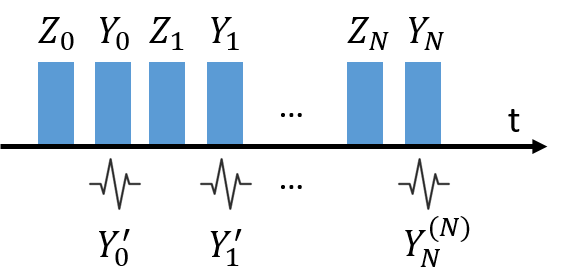}
\caption[Coherence times and sequence]{\textbf{a)} $T_1$ measurement with exponential fit. \textbf{b)} Spin-echo measurement, $T_{2E}$, with exponential fit. \textbf{c)} Sequence performed in the experiment. Blue boxes represent actual pulses. Logical Y and Z rotations are explicitly shown below the blue boxes~\cite{virtual-zgates}. Note that Z pulses do not correspond to any microwave pulse, instead subsequent pulses change rotation axis, indicated by a prime, $Y_N^{(N)}$.}
\label{fig:coh-times} 
\end{figure}

The experiment was realized in a dilution fridge with a base temperature of approximately $20$~mK. The qubit rotation pulses were defined by an arbitrary waveform generator (AWG) and then upconverted with a microwave signal generator to the gigahertz frequency range before being sent to the qubit/cavity system. The signal was low-pass filtered and attenuated by a total of 50dB before reaching the aluminum cavity. The input port of the cavity was undercoupled while the output port was overcoupled in order to maximize the readout signal amplitude. The outgoing signal was amplified by a cryogenic low noise amplifier and a second amplification stage at room temperature. The downconversion is performed with the same microwave generator as used in the upconversion of the measurement pulse, guaranteeing phase coherence in the downconversion process. The signal is read out in a digitizer, with a FPGA that demodulates and averages the results before sending the data to the main measurement computer.
\begin{figure}[t!]
\centering
\includegraphics[width=\linewidth]{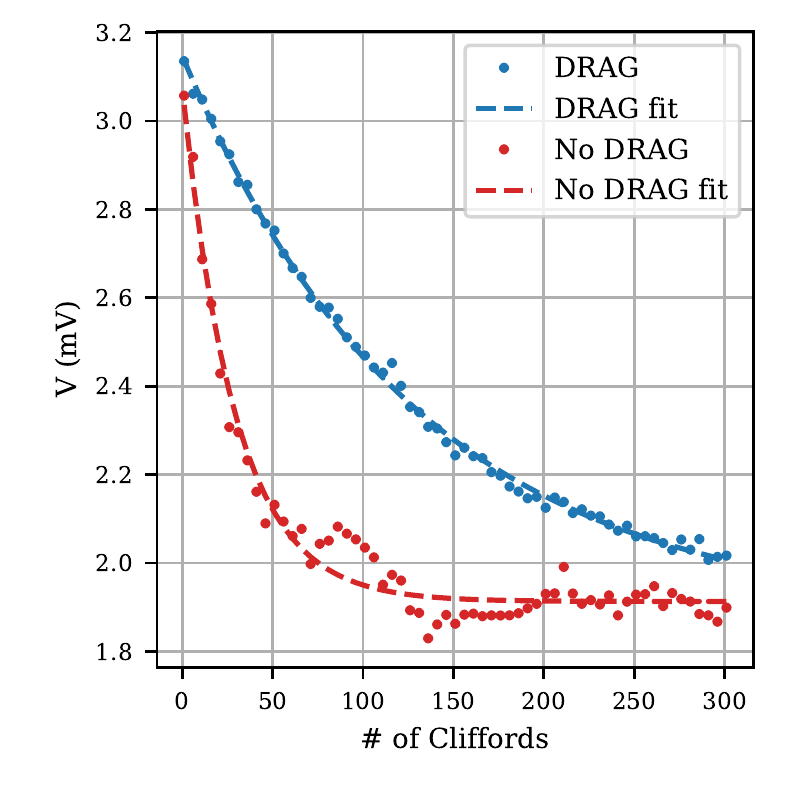}
\caption[Randomized Benchmarking]{Randomized benchmarking of the DRAG corrected pulses. The fit corresponds to the expression $Ap^n + B$, where $A$ and $B$ have dimensions of voltage, $n$ is the number of Clifford gates, and $p$ is the fidelity per gate. $\epsilon = 1-p$ is the error per gate.}
\label{fig:ran-ben} 
\end{figure}

Figure~\ref{fig:full-seq} shows the total pulse sequence, which includes preparation and measurement pulses in addition to the pulse sequence shown in the main text. The Y rotations are performed through microwave pulses at the qubit frequency while the Z rotations, as already stated, are phase changes in subsequent pulses. An example of the rotation angles for the $\relu(x)$ function in the 4-layer case is shown in Table~\ref{tab:rotation angles}. The readout consists of a cavity tone at the frequency of the cavity for the qubit in the $\ket{0}$ state. High/low transmission corresponds to the qubit being in the ground/excited state, assuming the system does not escape from the computational basis. Each data point requires around $50000$ measurements in order to average out the amplifier noise. A reset protocol that drives the qubit into the ground state is implemented prior to each individual sequence. This has two benefits. The first one allows us to start with a qubit state nearly polarized into the ground state. A second benefit is the reduction in the overall duration of the experiment, since the waiting time between individual measurements is not limited by the qubit relaxation time.

Both qubit and cavity pulses are generated at $70$~MHz and then upconverted to the gigahertz range. The qubit pulses are Gaussian pulses with a total duration of $21$~ns. A proper DRAG correction is performed with a resulting error per gate of $\epsilon = 0.01$ as shown in Fig.~\ref{fig:ran-ben}. The cavity pulse has a total length of around $2~\mu$s. The reset protocol consists of a pulse driving the qubit and a pulse driving the cavity mode, with a total duration of around $2~\mu$s.

\vfill

\end{document}